\documentclass[11pt]{article}

\usepackage[
  letterpaper,
  margin=1in,
  headheight=14pt
]{geometry}

\usepackage[T1]{fontenc}
\usepackage[utf8]{inputenc}
\usepackage{lmodern}
\usepackage[final]{microtype}
\usepackage{amsmath,amssymb,amsthm,mathtools,bm}
\usepackage{booktabs}
\usepackage{enumitem}
\setlist{itemsep=2pt,topsep=4pt}

\usepackage{xcolor}
\definecolor{linkblue}{RGB}{0,70,150}

\usepackage[
  colorlinks=true,
  linkcolor=linkblue,
  citecolor=linkblue,
  urlcolor=linkblue
]{hyperref}

\usepackage{titlesec}
\titleformat{\section}
  {\large\bfseries}
  {\thesection}{0.75em}{}
\titleformat{\subsection}
  {\normalsize\bfseries}
  {\thesubsection}{0.65em}{}
\titlespacing*{\section}{0pt}{2.3ex plus .4ex minus .2ex}{1.1ex}
\titlespacing*{\subsection}{0pt}{1.8ex plus .3ex minus .2ex}{0.8ex}

\usepackage{fancyhdr}
\usepackage{aliascnt}

\theoremstyle{plain}
\newtheorem{theorem}{Theorem}[section]

\newaliascnt{lemma}{theorem}
\newtheorem{lemma}[lemma]{Lemma}
\aliascntresetthe{lemma}

\newaliascnt{proposition}{theorem}
\newtheorem{proposition}[proposition]{Proposition}
\aliascntresetthe{proposition}

\newaliascnt{corollary}{theorem}

\aliascntresetthe{corollary}

\theoremstyle{definition}
\newaliascnt{assumption}{theorem}
\newtheorem{assumption}[assumption]{Assumption}
\aliascntresetthe{assumption}

\newaliascnt{definition}{theorem}
\newtheorem{definition}[definition]{Definition}
\aliascntresetthe{definition}

\newaliascnt{example}{theorem}
\newtheorem{example}[example]{Example}
\aliascntresetthe{example}

\theoremstyle{remark}
\newaliascnt{remark}{theorem}

\aliascntresetthe{remark}

\usepackage[nameinlink,capitalise,noabbrev]{cleveref}

\crefname{theorem}{theorem}{theorems}
\Crefname{theorem}{Theorem}{Theorems}
\crefname{lemma}{lemma}{lemmas}
\Crefname{lemma}{Lemma}{Lemmas}
\crefname{proposition}{proposition}{propositions}
\Crefname{proposition}{Proposition}{Propositions}
\crefname{corollary}{corollary}{corollaries}
\Crefname{corollary}{Corollary}{Corollaries}
\crefname{assumption}{assumption}{assumptions}
\Crefname{assumption}{Assumption}{Assumptions}
\crefname{definition}{definition}{definitions}
\Crefname{definition}{Definition}{Definitions}
\crefname{remark}{remark}{remarks}
\Crefname{remark}{Remark}{Remarks}
\crefname{example}{example}{examples}
\Crefname{example}{Example}{Examples}

\newcommand{\C}{\mathbb C}
\newcommand{\R}{\mathbb R}

\newcommand{\dd}{\,\mathrm d}
\newcommand{\eps}{\epsilon}

\newcommand{\range}{\operatorname{range}}
\newcommand{\diag}{\operatorname{diag}}

\newcommand{\norm}[1]{\left\lVert #1\right\rVert}

\newcommand{\ket}[1]{\left|#1\right\rangle}
\newcommand{\bra}[1]{\left\langle#1\right|}

\newcommand{\abs}[1]{\left|#1\right|}

\newcommand{\E}{\mathcal E}

\title{Toward Efficient End-to-End Quantum Elliptic PDE Solvers: a Multilevel Correction Algorithm for Direct Observable Estimations}
\author{Xiantao Li \\
Department of Mathematics, Pennsylvania State University\\
University Park, PA 16802, USA\\
\texttt{Xiantao.Li@psu.edu}}
\date{}
\numberwithin{equation}{section}
\hypersetup{pdftitle={Multilevel Correction for Quantum Estimation of Elliptic PDE Observables},pdfauthor={Xiantao Li}}

\begin{document}
\maketitle

\begin{abstract}
Spatial derivatives pose a basic challenge for quantum algorithms for
partial differential equations: their standard discretizations have
operator norms that grow as the mesh is refined, leading to large
block-encoding normalization factors. Even when preconditioning
improves access to the solution, estimating derivative-dependent
physical observables can reintroduce polynomial dependence on the
inverse mesh width. We develop a multilevel quantum algorithm for
estimating linear and quadratic observables of elliptic equations
that addresses this readout cost. The algorithm decomposes the output
into corrections across nested discretizations and realizes
fine--coarse cancellation coherently, before measurement.
A Ritz-complement factorization connects the $O(h^2)$ size of the
corrected response with a block encoding at the same scale, allowing
this decay to offset growing readout normalizations.
Under efficient quantum access to the corrected coordinates, input
data, and readouts, observables with readout scale $O(h^{-\chi})$,
$0\le\chi\le2$, can be estimated to additive accuracy $\epsilon$ at
cost $\widetilde O(\epsilon^{-1})$ using amplitude estimation or
$\widetilde O(\epsilon^{-2})$ using direct sampling, with only
polylogarithmic dependence on the inverse finest mesh width
$h_L^{-1}$.
We give explicit constructions of the required operator oracles for
one-dimensional piecewise linear finite elements with diffusion
piecewise constant on a fixed coarsest partition, and for Fourier
spectral hierarchies. Output-specific bias estimates connect the
discrete results to continuum accuracy.
\end{abstract}



\section{Introduction}
\label{sec:introduction}

Quantum algorithms for partial differential equations are usually introduced through
the linear system produced by a spatial discretization.  For elliptic problems,
quantum linear-system algorithms (QLSAs) can prepare a state proportional to the
discrete solution, and quantum preconditioners can substantially reduce the cost of
that step.  A scientific computation, however, usually seeks a physical output:
an average, a flux, or an energy in a region.  Estimating that output is part of the
algorithm, and its cost can remain sensitive to the mesh even after the solve has
been preconditioned.  This paper develops a multilevel correction method that brings
the output into the design of the quantum computation.

Consider the model problem
\[
 -\nabla\cdot(a(x)\nabla u(x))=f(x)\quad\text{in }\Omega,
 \qquad u|_{\partial\Omega}=0.
\]
Finite element discretization gives a sparse positive-definite system.  QLSAs,
inverse block encodings, and singular value transformation provide quantum tools
for such systems
\cite{HarrowHassidimLloyd2009,ChildsKothariSomma2017,GilyenSuLowWiebe2019,
AnLin2022,CostaAnSandersSuBabbushBerry2022}.
These tools have been applied to Poisson equations and finite element and spectral
PDE solvers
\cite{CaoEtAl2013,MontanaroPallister2016,ChildsLiuOstrander2021,OrsucciDunjko2021}.
Their efficiency depends on both conditioning and access to the data and matrices.
For example, positive-definite factorizations can reduce the elliptic solve scale
from $h^{-2}$ to $h^{-1}$ under appropriate structural assumptions
\cite{OrsucciDunjko2021}.

Classical multigrid and subspace correction use nested finite element spaces to
overcome mesh-dependent conditioning
\cite{BramblePasciakXu1990,Xu1992,Hackbusch1985,TrottenbergOosterleeSchuller2001}.
Quantum realizations of these ideas include BPX and domain-decomposition
preconditioners and unitary formulations of iterative methods
\cite{DeimlPeterseim2025,FressartNowakSpillane2026,JinLiuMaYu2025Preconditioning,
YangYuZhang2025}.
In particular, Deiml and Peterseim already prove end-to-end complexity
$\widetilde O(\epsilon^{-1})$ for suitable linear functionals, with explicit
finite element circuits \cite{DeimlPeterseim2025}.  Our purpose is to develop a
corrected inverse-response operator that can be composed with additional readouts,
including genuinely quadratic outputs.  The common $\epsilon^{-1}$ dependence
for their linear-functional setting is not a new result of this paper.
More specifically, \cite[Theorem~6.5]{DeimlPeterseim2025} estimates a linear
functional using a BPX-preconditioned finite element system.  The additional
contributions here are: (i) a composable encoding of the corrected inverse
$h_\ell^{-2}D_\ell$; (ii) compatible quadratic level identities, including
their cross terms; and (iii) a Ritz-detail construction that places the
$h_\ell^2$ cancellation in the operator before estimation.  The Schur
identity and the underlying finite element stability are classical; their
use for this output-level quantum construction is the contribution.

To explain the readout issue, write the discrete quantities of interest as
\[
 Q_h^{\rm lin}=\bm g_h^\dagger\bm x_h,
 \qquad Q_h^{\rm quad}=\bm x_h^\dagger M_h\bm x_h.
\]
We use coordinates whose Euclidean norm is equivalent to the physical $L^2$ norm.
In these coordinates, a volume-gradient readout can carry scale $h^{-1}$,
and a regional-energy matrix can carry scale $h^{-2}$.  We denote the effective
readout scale by $O(h^{-\chi})$, including input norms and the normalization of
the chosen observable block encoding.  If an $O(1)$-normalized inverse response
is available, a direct estimator then has amplitude-estimation cost proportional
to $\epsilon^{-1}h^{-\chi}$; ordinary sampling has cost proportional to
$\epsilon^{-2}h^{-2\chi}$.  These are upper bounds for that implementation,
not lower bounds for every representation of the same physical output. 

BPX and domain-decomposition constructions also provide access
to preconditioned gradient maps
\cite{DeimlPeterseim2025,FressartNowakSpillane2026}.
For outputs admitting an efficiently accessible bounded
representation in these coordinates, such maps can also avoid
mesh-growing readout normalizations.
Our comparison concerns the specified direct estimators.
The contribution here is a composable corrected inverse at
the $h_\ell^2$ scale, together with multilevel estimation
identities and error bounds for compatible linear and
quadratic readouts.

This distinction belongs to the end-to-end viewpoint of
\cite{MontanaroPallister2016}: data preparation, solution access, and inference
must all be counted.  Structured load and readout data can admit efficient state
preparation \cite{GroverRudolph2002,KayeMosca2001,DeimlPeterseim2025}; arbitrary
tables need not.  A normalized solution state also omits its physical norm, which
may require a separate estimation step \cite{dalzell2024shortcut}.  The method
below estimates the output directly from the load and so avoids this intermediate
solution-state normalization.

Our starting point is the multilevel telescoping identity
\[
 Q_L=Q_0+\sum_{\ell=1}^L(Q_\ell-Q_{\ell-1}).
\]
In classical multilevel Monte Carlo (MLMC), coupled fine and coarse samples make
the differences less variable than the individual outputs
\cite{Giles2008MLMC,Giles2015MLMC}.  Here the PDE and its level outputs are
deterministic; the randomness comes from quantum measurement.  The analogous
gain is obtained by representing the difference inside the measured operator.
Subtracting two independently estimated outputs would retain their full scales.

Let $B_\ell=A_\ell^{-1}$ and let $P_\ell$ be the consistently scaled Galerkin
prolongation.  The corrected discrete Green operator is
\[
 D_\ell=B_\ell-P_\ell B_{\ell-1}P_\ell^\dagger.
\]
Standard hierarchical block elimination gives
\[
 D_\ell=J_\ell\Sigma_\ell^{-1}J_\ell^\dagger,
\]
where $J_\ell$ maps detail coordinates into the energy-orthogonal complement of
the coarse space and $\Sigma_\ell$ is its energy Gram matrix.  The new use of
this factorization is to expose the fine--coarse cancellation in a quantum
observable.  Dual approximation yields $\|D_\ell\|=O(h_\ell^2)$ in the full
regularity case.  When stable detail coordinates admit block encodings at their
natural $O(h_\ell)$ scale, the factorization achieves normalization
$O(h_\ell^2)$ as well.  Both linear and quadratic level estimators then have
effective scale $O(h_\ell^{2-\chi})$.

We prove these statements separately at the algebraic, analytic, and circuit
levels.  The access requirements are verified for one-dimensional midpoint
$P_1$ details with aligned piecewise-constant diffusion, and for nested sine and
Fourier spectral hierarchies.  In the spectral case, a transform, a band test,
and reciprocal evaluation implement $D_\ell$ directly.  General finite element
meshes remain conditional on the stated coordinate access; classical stability
or localization results do not by themselves supply that access.

Write $h=h_L$ for the finest mesh width and let $\tau_L$
be the gate cost of one call to the direct finest-level
estimator, including input preparation and operator access.
For continuum accuracy, choose the smallest $L$ for which
the discretization bias is at most $\epsilon/2$.
An algebraic bias bound $O(h_L^p)$, with fixed $p>0$,
then gives $L=O(1+\log\epsilon^{-1})$.

\begin{table}[htp]
\centering\small
\caption{Total gate costs for $0\le\chi\le2$.
The multilevel bounds assume the access hypotheses of
\Cref{thm:overall-complexity}.
The notation $\widetilde O$ suppresses polylogarithmic
factors in $h^{-1}$ and $\epsilon^{-1}$.}
\label{tab:intro-comparison}
\begin{tabular}{@{}lcc@{}}
\toprule
Estimator & Amplitude estimation & Direct sampling\\
\midrule
Single level
 & $\widetilde O(\tau_L h^{-\chi}\epsilon^{-1})$
 & $\widetilde O(\tau_L h^{-2\chi}\epsilon^{-2})$\\
Multilevel
 & $\widetilde O(\epsilon^{-1})$
 & $\widetilde O(\epsilon^{-2})$\\
\bottomrule
\end{tabular}
\end{table}

\Cref{tab:intro-comparison} compares the total gate costs of the
direct and multilevel estimators. Under the stated access
hypotheses, the multilevel construction suppresses polynomial
dependence on the inverse finest mesh width $h^{-1}$ for
$0\le\chi\le2$. For $\chi>0$, the direct readout factors grow
exponentially with the refinement level $L$, whereas the
remaining mesh dependence is only polynomial in $L$.
The tolerance $\epsilon$ initially concerns the discrete output.
For a continuum target, an algebraic output-bias bound permits
$L=O(1+\log\epsilon^{-1})$, so the level factors contribute only
polylogarithmic dependence on the inverse tolerance.
We give representative bias bounds below.
These guarantees concern estimation of specified outputs
under the stated data-access assumptions.

The paper proceeds from finite element coordinates and outputs
(\Cref{sec:elliptic}), through quantum primitives and multilevel allocation
(\Cref{sec:quantum-primitives,sec:mlmc}), to the corrected Green construction,
its explicit examples, and the complexity theorem
(\Cref{sec:two-level-estimators}).  \Cref{sec:scaled-two-level-qlsa} shows how
hierarchical elimination of a coupled two-level QLSA leads back to the same Ritz
complement.

\subsection*{Notation}
Vectors are bold, $\dagger$ denotes the conjugate transpose, and $\|\cdot\|_2$
denotes the Euclidean or induced spectral norm.  For a full-rank matrix,
$\kappa(A)=\sigma_{\max}(A)/\sigma_{\min}(A)$.  We use standard Sobolev spaces
$L^2(\Omega)$, $H^1(\Omega)$, and $H_0^1(\Omega)$, and the energy norm induced
by the elliptic bilinear form.  The level index is $\ell$, the finest level is
$L$, and $h=h_L$.  The spatial dimension and the domain are fixed; constants may
depend on ellipticity, shape regularity, and the approximation constants, but not
on $h$.  The notation $\widetilde O$ suppresses polylogarithmic factors in mesh
resolution, precision, and inverse failure probability.

Bra-ket notation denotes normalized quantum states.  For a nonzero vector
$\bm v\in\mathbb C^N$, amplitude encoding means
\begin{equation}
 \ket{v}=\frac{1}{\|\bm v\|_2}\sum_{j=0}^{N-1}v_j\ket{j}.
\label{q-encode}
\end{equation}
The register has $\lceil\log_2N\rceil$ qubits; unused entries are padded with
zeros.  Physical coefficient vectors remain unnormalized unless a ket is used.

\section{Finite element discretizations and quantities of interest}
\label{sec:elliptic}

Let \(\Omega\subset\R^d\) be a bounded polyhedral domain.  We consider the scalar elliptic model problem
\begin{equation}
-\nabla\cdot(a(x)\nabla u(x))=f(x),\qquad u|_{\partial\Omega}=0,
\label{eq:pde}
\end{equation}
where $f\in L^2(\Omega)$ is fixed and
$a\in L^\infty(\Omega;\mathbb R^{d\times d})$
is symmetric and uniformly elliptic almost everywhere:
\[
a_{\min}\abs{\bm \xi}^2\le \bm \xi^\dagger a(x)\bm \xi\le a_{\max}\abs{\bm \xi}^2 .
\]
The weak formulation corresponds to finding \(u\in H_0^1(\Omega)\) such that
\begin{equation}
a(u,v):=\int_\Omega \nabla v(x)^\dagger a(x)\nabla u(x)\dd x
=
\int_\Omega f(x)v(x)\dd x=:F(v),
\label{eq:weak}
\end{equation}
for all \(v\in H_0^1(\Omega)\).  The Galerkin discretization and the finite element error estimates used below are standard; see, for example, \cite{BrennerScott2008}. The positive-definite quadratic form in \eqref{eq:weak} also induces an energy norm, $\norm{u}_a= \sqrt{a(u,u)},$ which is quite useful in the analysis. 

For the main analysis, we adopt the standard approximation property
underlying the Aubin--Nitsche estimate \cite{BrennerScott2008}.
For each $g\in L^2(\Omega)$, let $\phi_g\in H_0^1(\Omega)$ solve
\[
 a(v,\phi_g)=(v,g)_{L^2(\Omega)}
 \qquad\forall v\in H_0^1(\Omega).
\]
We assume that
\begin{equation}
 \inf_{q\in V_{\ell-1}}\|\phi_g-q\|_a
 \le C_{\rm dual}h_{\ell-1}\|g\|_{L^2(\Omega)},
\label{eq:dual-approximation}
\end{equation}
with $C_{\rm dual}$ independent of the level. This property follows
from $H^2$ regularity and standard finite element approximation.
The approximation formulation also accommodates mesh-aligned
interface problems with the corresponding piecewise regularity.

Meanwhile for the numerical discretization, we consider a sequence of nested trial spaces \cite{BriggsHensonMcCormick2000}:
\begin{equation}\label{nested}
    V_0\subset V_1\subset\cdots\subset V_L\subset H_0^1(\Omega).
\end{equation}
We consider conforming $P_1$ or $Q_1$ finite element spaces on a nested,
shape-regular, quasi-uniform dyadic mesh hierarchy with
\[
h_\ell=h_{\ell-1}/2,\qquad h=h_L.
\]
The finest mesh width is chosen from an error estimate for the requested
quantity of interest.  If its bias is $O(h_L^p)$, with fixed $p>0$, then
$h_L=O(\epsilon^{1/p})$ and $L=O(\log\epsilon^{-1})$ suffice.
\Cref{subsec:qoi} distinguishes the rates for representative outputs.

\smallskip 

Let $N_\ell=\dim(V_\ell)$ and collect the nodal basis functions in
$\Phi_\ell(x)=(\varphi_1^{(\ell)}(x),\ldots,
\varphi_{N_\ell}^{(\ell)}(x))^T$.  For convenience, we consider the
standard Lagrange nodal basis, normalized by
\begin{equation}
\varphi_i^{(\ell)}(x_j^{(\ell)})=\delta_{ij}.
\label{eq:nodal-basis-property}
\end{equation} The Galerkin solution expressed in $V_\ell$ in terms of the shape functions $\Phi_\ell$ has coefficient vector \(\bm x_\ell^{\rm raw}\in\R^{N_\ell}\) satisfying
\begin{equation}
A_\ell^{\rm raw}\bm x_\ell^{\rm raw}=\bm b_\ell^{\rm raw},
\qquad
(A_\ell^{\rm raw})_{ij}=a(\varphi_j^{(\ell)},\varphi_i^{(\ell)}),
\qquad
(\bm b_\ell^{\rm raw})_i=F(\varphi_i^{(\ell)}).
\label{eq:discrete-system-raw}
\end{equation}
All matrices are real at this stage. When  used in quantum circuits, they are viewed as complex matrices without changing the notation.

\subsection{Quantum-compatible reduced scaling}
\label{subsec:response-coordinates}

To compare circuit normalizations across levels, we remove the grid-volume
factor from nodal vectors.  For a finite element function $v_\ell$ with nodal
vector $\bm v_\ell^{\rm raw}$, quasi-uniformity gives
\begin{equation}
 \|v_\ell\|_{L^2(\Omega)}\simeq
 h_\ell^{d/2}\|\bm v_\ell^{\rm raw}\|_2.
\label{eq:nodal-L2-equivalence}
\end{equation}
The load has the dual bound
$\|\bm b_\ell^{\rm raw}\|_2\le Ch_\ell^{d/2}\|f\|_{L^2}$.
These relations suggest the reduced coordinates
\begin{equation}
 \bm x_\ell=h_\ell^{d/2}\bm x_\ell^{\rm raw},\qquad
 \bm b_\ell=h_\ell^{-d/2}\bm b_\ell^{\rm raw},\qquad
 A_\ell=h_\ell^{-d}A_\ell^{\rm raw},
\label{eq:response-coordinate-def}
\end{equation}
which satisfy
\begin{equation}
 A_\ell\bm x_\ell=\bm b_\ell.
\label{eq:response-discrete-system}
\end{equation}
Only known scalar factors are used; no mass-matrix square root is needed.
The Euclidean norm now has the physical $L^2$ scale, up to uniform constants.

For $u_\ell=h_\ell^{-d/2}\Phi_\ell^\dagger\bm x_\ell$ and
$v_\ell=h_\ell^{-d/2}\Phi_\ell^\dagger\bm y_\ell$, the energy pairing is
exactly preserved:
\begin{equation}
 a(u_\ell,v_\ell)
 =(\bm y_\ell^{\rm raw})^\dagger A_\ell^{\rm raw}\bm x_\ell^{\rm raw}
 =\bm y_\ell^\dagger A_\ell\bm x_\ell.
\label{eq:energy-preserved-amplitude-coordinates}
\end{equation}
Similarly, $F(v_\ell)=\bm y_\ell^\dagger\bm b_\ell$.
Thus the rescaling changes neither the Galerkin problem nor its physical
outputs; it avoids introducing artificial dimension-dependent factors into the
quantum cost.

The transfer must use the same convention.  If $P_\ell^{\rm raw}$ is ordinary
nodal prolongation, set
\begin{equation}
 P_\ell=h_\ell^{d/2}P_\ell^{\rm raw}h_{\ell-1}^{-d/2}.
\label{eq:reduced-prolongation}
\end{equation}
Then $P_\ell\bm y_{\ell-1}$ represents the same coarse function on the fine
mesh:
\begin{equation}
 h_\ell^{-d/2}\Phi_\ell(x)^\dagger P_\ell\bm y_{\ell-1}
 =h_{\ell-1}^{-d/2}\Phi_{\ell-1}(x)^\dagger\bm y_{\ell-1}.
\label{eq:P-function-meaning}
\end{equation}
Consequently, the Galerkin identity retains its familiar form,
\begin{equation}
 A_{\ell-1}=P_\ell^\dagger A_\ell P_\ell.
\label{eq:reduced-galerkin-identities}
\end{equation}
Mass equivalence and the finite element inverse inequality
\cite{BrennerScott2008} give the useful bounds
\begin{equation}
 \|A_\ell\|_2=O(h_\ell^{-2}),\qquad
 \|A_\ell^{-1}\|_2=O(1),\qquad \|P_\ell\|_2=O(1).
\label{eq:A-response-norm-scale}
\end{equation}
Here the inverse bound also uses coercivity and the Poincar\'e inequality.
Unless stated otherwise, all subsequent matrices and vectors use these reduced
coordinates.

\subsection{Transfer operators and multigrid operations}
\label{subsec:transfer-residual}

In addition to the prolongation $P_\ell$, we choose the restriction as the adjoint transfer \(P_\ell^\dagger\) as standard  in Galerkin multigrid  \cite{Hackbusch1985,BriggsHensonMcCormick2000,TrottenbergOosterleeSchuller2001,Xu1992}. Due to the Galerkin projection on the nested finite element spaces, the load vector and stiffness matrices are related between two successive levels of discretization, 
\begin{equation}
A_{\ell-1}=P_\ell^\dagger A_\ell P_\ell,
\qquad
\bm b_{\ell-1}=P_\ell^\dagger\bm b_\ell.
\label{eq:galerkin-consistency}
\end{equation}
In addition,  we define the discrete Green operator at the same level as
\begin{equation}
B_\ell:=A_\ell^{-1}.
\label{eq:B-def}
\end{equation}
We have $
\bm x_\ell=B_\ell\bm b_\ell.
$

At the finest level we frequently abbreviate
\[
A:=A_L,
\qquad
B:=B_L,
\qquad
\bm b:=\bm b_L,
\qquad
\bm x:=\bm x_L,
\]
which are the quantities one would work with in a single-level method.

\subsection{Linear and quadratic quantities of interest}
\label{subsec:qoi}

We consider fixed linear functionals and Hermitian quadratic forms restricted to
the nested trial spaces.  Their reduced-coordinate representations are
\begin{equation}
 \mathcal L_\ell(\bm x_\ell)
 =\bm g_\ell^\dagger\bm x_\ell
 =\bm g_\ell^\dagger B_\ell\bm b_\ell,
\label{eq:linear-qoi}
\end{equation}
and
\begin{equation}
 \mathcal Q_\ell(\bm x_\ell)
 =\bm x_\ell^\dagger M_\ell\bm x_\ell
 =\bm b_\ell^\dagger B_\ell^\dagger M_\ell B_\ell\bm b_\ell.
\label{eq:quadratic-qoi}
\end{equation}
If $\bm g_\ell^{\rm raw}$ and $M_\ell^{\rm raw}$ act on raw nodal
vectors, their reduced forms are
\begin{equation}
 \bm g_\ell=h_\ell^{-d/2}\bm g_\ell^{\rm raw},\qquad
 M_\ell=h_\ell^{-d}M_\ell^{\rm raw}.
\label{eq:reduced-readout-scaling}
\end{equation}
These choices preserve both physical output values under
\eqref{eq:response-coordinate-def}.

Examples of linear outputs are volume averages $\int_\Omega\psi u_\ell$ and
volume-averaged fluxes $\int_\Omega\bm\psi\cdot a\nabla u_\ell$, with fixed
$L^2$ weights.  The first has a bounded reduced readout norm.  For the second,
the inverse estimate $\|\nabla v_\ell\|_{L^2}\le Ch_\ell^{-1}\|v_\ell\|_{L^2}$
gives the upper bound $\|\bm g_\ell\|_2=O(h_\ell^{-1})$.
Additional structure in the weight can yield a smaller effective scale.

Representative quadratic outputs are weighted mass and regional energy,
\[
 \int_\Omega\omega(x)|u_\ell(x)|^2\,dx,
 \qquad
 \int_\Omega\omega(x)\nabla u_\ell(x)^\dagger a(x)\nabla u_\ell(x)\,dx,
\]
 where \(\omega\) is real-valued, fixed, and bounded.  For the explicit finite element
regional-energy circuit, take its region to be a union of coarsest-mesh
elements and the coefficient integrals to be reversibly accessible.  Then
the indicator is decided from the coarse element label, and the usual
fixed-size element stiffness assembly gives each nonzero matrix entry.
Arbitrary cut-cell geometry requires an additional integration oracle.  Their matrices have norm bounds $O(1)$ and $O(h_\ell^{-2})$,
respectively, and sparse local encodings attain these upper scales when the
element integrals are efficiently accessible.  Regional energy is a genuinely
quadratic example.  Total energy, with $\omega=1$, is the special case
$\bm x_\ell^\dagger A_\ell\bm x_\ell=\bm b_\ell^\dagger\bm x_\ell$ and
therefore reduces to the linear compliance output.

Since the same continuous forms are restricted to nested spaces, their
representations satisfy
\begin{equation}
 \bm g_{\ell-1}=P_\ell^\dagger\bm g_\ell,
\label{eq:g-compatible}
\end{equation}
and
\begin{equation}
 M_{\ell-1}=P_\ell^\dagger M_\ell P_\ell.
\label{eq:M-compatible}
\end{equation}
Exact integration, or a consistently inherited discrete form, preserves these
identities.  Independent quadrature or level-dependent reconstructions may add
consistency remainders, which must be estimated separately.  Boundary fluxes and
tractions likewise require a specified trace or flux reconstruction, together
with its normalization and bias estimates; the volume inverse estimate alone
does not assign them readout order one.

\begin{definition}[Effective readout order]
\label{def:effective-readout-orders}
A linear readout family has admissible order $\chi_g\ge0$ if
\begin{equation}
 \|\bm g_\ell\|_2=O(h_\ell^{-\chi_g}),
\label{eq:g-effective-scale}
\end{equation}
with its known scalar normalization included.  Relative to a specified family
of block encodings, a quadratic readout has admissible order $\chi_M\ge0$ if
\begin{equation}
 \alpha_{M_\ell}=O(h_\ell^{-\chi_M}).
\label{eq:M-effective-scale}
\end{equation}
We use $\chi$ for the relevant order.  The theorem needs an upper bound, not the
smallest possible exponent.  For quadratic outputs, a spectral-norm estimate
alone is insufficient: \eqref{eq:M-effective-scale} refers to the circuit that
will actually be used.
\end{definition}

Under \eqref{eq:dual-approximation}, C\'ea's lemma and the
Aubin--Nitsche argument give
$\|u-u_\ell\|_{H^1}=O(h_\ell)$ and
$\|u-u_\ell\|_{L^2}=O(h_\ell^2)$.
Volume averages with fixed $L^2$ weights and weighted
$L^2$ mass therefore have bias $O(h_\ell^2)$.
Volume fluxes with fixed $L^2$ weights and regional energies
have the sufficient bias bound $O(h_\ell)$.
These bounds follow from Cauchy--Schwarz and the uniform
energy bound on the Galerkin solutions.
For continuum accuracy $\epsilon$, reserve $\epsilon/2$
for this bias and apply \Cref{thm:overall-complexity}
with the remaining tolerance.

\subsection{State normalization and physical prefactors}
\label{subsec:normalized-prefactors}

The reduced scaling removes the common grid-volume factor,
but the norms of $\bm b_\ell$ and $\bm g_\ell$ still enter
the physical output. The readout norm may retain mesh
dependence. The quantum circuit prepares normalized states,
\begin{equation}
|b_\ell\rangle:=\frac{\bm b_\ell}{\|\bm b_\ell\|_2},
\qquad
|g_\ell\rangle:=\frac{\bm g_\ell}{\|\bm g_\ell\|_2},
\label{eq:normalized-bg-states}
\end{equation}
whenever the corresponding vector is nonzero.  Therefore the physical linear QoI is recovered from a normalized overlap as
\begin{equation}
\mathcal L_\ell
=
\|\bm g_\ell\|_2\,\|\bm b_\ell\|_2
\langle g_\ell|B_\ell|b_\ell\rangle .
\label{eq:linear-normalized-prefactor}
\end{equation}
Likewise, the physical quadratic QoI is
\begin{equation}
\mathcal Q_\ell
=
\|\bm b_\ell\|_2^2
\langle b_\ell|B_\ell^\dagger M_\ell B_\ell|b_\ell\rangle .
\label{eq:quadratic-normalized-prefactor}
\end{equation}
The scalar factors \(\|\bm b_\ell\|_2\) and \(\|\bm g_\ell\|_2\) are not obtained by sampling the solution state.  They are known from the input preparation procedure or are classically computable from the data model.  In  observable estimation, these factors can be kept outside the expectation or absorbed into the effective level normalization. Another important point is that while we can assume that $\norm{\bm b_\ell}$ is accessible, the norm of $\bm x_\ell$ usually has to be estimated separately from a QLSA approach \cite{dalzell2024shortcut}.

\subsection{Preparation of the load vector}
\label{subsec:rhs-preparation}

State preparation is an input requirement, even when a discretization is
sparse.  The finite element load entries are
\begin{equation}
 (\bm b_\ell^{\rm raw})_j
 =\int_\Omega f(x)\varphi_j^{(\ell)}(x)\,dx.
\label{eq:load-integrals}
\end{equation}
For a padded binary index $j=yz$, where $y$ is a prefix, the binary-tree
preparation of $\ket{b_\ell}$ needs the nonnegative weights
\begin{equation}
 p_\ell(y)=\sum_z |(\bm b_\ell)_{yz}|^2,
 \qquad p_\ell(\varnothing)=\|\bm b_\ell\|_2^2,
\label{eq:squared-prefix-weights}
\end{equation}
and the phases of the entries.  Conditional rotations use the ratios
$p_\ell(y0)/p_\ell(y)$ and $p_\ell(y1)/p_\ell(y)$.  Zero-weight branches
are skipped.  The constructions of
\cite{GroverRudolph2002,KayeMosca2001} use $O(\log N_\ell)$ such stages.
Appendix~A of \cite{DeimlPeterseim2025} discusses this tree in the finite
element setting.  We assume that these weights and phases permit preparation of
$\ket{b_\ell}$ to state-vector error $O(\eta)$ and evaluation of
$\|\bm b_\ell\|_2$ to relative error $O(\eta)$, at cost
$\operatorname{polylog}(N_\ell,\eta^{-1})$ on each level.
The same requirements apply to $\ket{g_\ell}$ and its norm.

\section{Quantum primitives}
\label{sec:quantum-primitives}

Block encodings connect the matrix identities below to quantum circuits.  We
recall the standard rules and track their normalizations, since amplitude
estimation costs scale linearly in the normalization and direct-sampling
variance scales quadratically.

\begin{definition}[Block encoding]
Let \(M\in\C^{m\times n}\).  A unitary \(U_M\) is an
\((\alpha_M,n_a,\eps_M)\)-block encoding of \(M\) if
\[
\norm{
M-\alpha_M(\bra{0}^{\otimes n_a}\otimes I_m)U_M
(\ket{0}^{\otimes n_a}\otimes I_n)
}_2
\le \eps_M .
\]
Here \(n_a\) is the number of ancilla qubits.  For rectangular matrices, the input and output
registers are understood to be embedded into a common unitary space by padding if needed.
The scalar \(\alpha_M\) is the block-encoding normalization. In the exact case it must
satisfy \(\alpha_M\ge\norm{M}_2\).  Following
\cite[Definition~3.2]{DeimlPeterseim2025}, one may also track the subnormalization
\[
\widetilde\alpha_M:=\alpha_M/\norm{M}_2,
\]
which becomes important when block encodings are composed.
\end{definition}

\begin{lemma}[Sparse block encodings]
\label{lem:sparse-block-encoding}
Let \(M\in\C^{m\times n}\) be \(s_r\)-row-sparse and \(s_c\)-column-sparse.  Suppose
sparse-access oracles return the nonzero positions and entries, and suppose
\[
|M_{ij}|\le M_{\max}.
\]
Then \(M\) admits a block encoding with normalization
\[
\alpha_M=O\!\left(\sqrt{s_rs_c}\,M_{\max}\right)
\]
using \(O(1)\) sparse-oracle calls and polylogarithmic additional gates, up to
entry-precision factors.
\end{lemma}
The normalization in \Cref{lem:sparse-block-encoding} is consistent with the spectral
norm bound
\[
\norm{M}_2
\le
\sqrt{\norm{M}_1\norm{M}_\infty}
\le
\sqrt{s_c s_r}\,M_{\max}.
\]

\begin{lemma}[Block-encoding algebra]
\label{lem:block-algebra}
Let \(U_A\) and \(U_B\) block encode compatible matrices \(A\) and \(B\) with
normalizations \(\alpha_A\) and \(\alpha_B\).  Then:
\begin{enumerate}[label=(\roman*),leftmargin=2em]
\item \(A^\dagger\) is block encoded with normalization \(\alpha_A\).
\item \(AB\) is block encoded with normalization \(\alpha_A\alpha_B\).
\item \(\mu_AA+\mu_BB\) is block encoded by LCU with normalization
\[
|\mu_A|\alpha_A+|\mu_B|\alpha_B.
\]
\item The block diagonal matrix \(A\oplus B\) can be block encoded with normalization
\[
\max\{\alpha_A,\alpha_B\}.
\]
\end{enumerate}
\end{lemma}
The product and LCU constructions appear as Lemmas~53 and~52 of Gily\'en--Su--Low--Wiebe \cite{GilyenSuLowWiebe2019}.  Deiml and Peterseim give a form adapted to finite element block encodings  \cite[Proposition~3.3]{DeimlPeterseim2025}.

The primitives in (\emph{i}) and (\emph{ii}) build a stiffness block encoding from a gradient matrix. The normalization factors $\alpha$ will directly impact the cost of estimating QoIs.

\begin{lemma}[Inverse block encoding]
\label{lem:inverse-block-encoding}
Suppose $U_A$ exactly block encodes $A$ with normalization $\alpha_A$,
and every nonzero singular value of $A$ is at least $\sigma_{\min}>0$.
For $0<\varepsilon<\sigma_{\min}^{-1}$, QSVT constructs a block encoding
of $A^+$ with normalization $O(\sigma_{\min}^{-1})$ and absolute operator
error at most $\varepsilon$, using

$$
 O\!\left(
 \frac{\alpha_A}{\sigma_{\min}}
 \log\frac{2}{\sigma_{\min}\varepsilon}
 \right)
$$

queries to $U_A$ and $U_A^\dagger$
\cite[Theorem~41]{GilyenSuLowWiebe2019}.
\end{lemma}

\begin{lemma}[Estimating linear and quadratic forms]
\label{lem:rall-estimation}
Let \(B\) have a block encoding with normalization \(\alpha_B\), and let \(M\) have a
Hermitian block encoding with normalization \(\alpha_M\).  Suppose \(\ket{b}\) and, for
linear forms, \(\ket{g}\), can be prepared.
\begin{enumerate}[label=(\roman*),leftmargin=2em]
\item The normalized overlap \(\bra{g}B\ket{b}\) can be estimated with
amplitude-estimation cost proportional to
\[
\alpha_B/\eps.
\]
The final value \(\bm g^\dagger B\bm b\) is then obtained by multiplying by
\(
\norm{\bm g}_2\norm{\bm b}_2.
\)
\item The normalized quadratic form \(\bra{b}B^\dagger M B\ket{b}\) can be estimated with
amplitude-estimation cost proportional to
\[
\alpha_B^2\alpha_M/\eps.
\]
The final value \(\bm b^\dagger B^\dagger M B\bm b\) is then obtained by multiplying
by
\(
\norm{\bm b}_2^2.
\)
\end{enumerate}
\end{lemma}

\section{The multilevel principle for quantum observable estimation}
\label{sec:mlmc}

Recall that multigrid methods are often constructed from a hierarchy of nested subspaces
\[
V_0\subset V_1\subset\cdots\subset V_L,
\qquad
h_\ell=h_{\ell-1}/2,
\qquad
h=h_L.
\]
Such a hierarchy naturally gives a sequence of level-dependent outputs \(Q_\ell\).  The
fundamental multilevel identity is the telescoping sum
\begin{equation}
Q_L=Q_0+\sum_{\ell=1}^L \Delta_\ell,
\qquad
\Delta_\ell:=Q_\ell-Q_{\ell-1}.
\label{eq:ml-telescoping}
\end{equation}
Classical MLMC combines \eqref{eq:ml-telescoping} with a coupling of \(Q_\ell\) and
\(Q_{\ell-1}\), so that \(\Delta_\ell\) has much smaller variance than \(Q_\ell\) itself
\cite{Giles2008MLMC,Giles2015MLMC}.

The same principle is useful in quantum observable estimation, but the relevant scale is
not only a classical variance.  Quantum readout has two standard regimes.  With direct
projective measurement, a level-\(\ell\) observable with effective scale \(\alpha_\ell\)
gives a bounded random variable whose variance is at most \(O(\alpha_\ell^2)\), using a Hadamard-test realization of the encoded overlap.   Thus
estimating a level contribution to accuracy \(\epsilon_\ell\) by direct sampling costs
\[
\widetilde O\!\left(
\frac{\tau_\ell\alpha_\ell^2}{\epsilon_\ell^2}
\right),
\]
where \(\tau_\ell\) is the cost of one level-\(\ell\) circuit call.  Amplitude
estimation \cite{Rall2020}, via a coherent circuit, improves the precision scaling to 
\[
\widetilde O\!\left(
\frac{\tau_\ell\alpha_\ell}{\epsilon_\ell} 
\right).
\]
  In both regimes,  the quantum analogue of MLMC variance
reduction is to make the level operator itself smaller before measurement.

Estimating \(Q_\ell\) and \(Q_{\ell-1}\) separately and subtracting the two estimates
classically does not lower the quantum measurement cost, because each circuit still sees
the full single-level observable scale.  The cancellation must be represented inside the
operator or state entering the quantum circuit.  The rest of the paper implements this
idea for elliptic finite element quantities of interest.

In a single-level circuit, an $O(1)$ physical output can have normalized
signal of size $O(h_\ell^\chi)$.  The multilevel circuit instead normalizes
each correction at its upper scale $O(h_\ell^{2-\chi})$, relaxing the precision required of the measured amplitude.  The actual correction
may be much smaller or zero; the proof needs a bound on its measurement
range, not a positive lower bound on its signal.

We use the following allocation estimate repeatedly.  For notational convenience, set
\(\Delta_0:=Q_0\).

\begin{lemma}[MLMC allocation for quantum observable estimation]
\label{lem:mlmc-allocation}
Let
\[
Q_L=\sum_{\ell=0}^L\Delta_\ell
\]
be a multilevel decomposition.  Let \(\tau_\ell\) denote the cost of one level-\(\ell\)
circuit call, including state preparation, block encodings, and any levelwise
normalization subroutines used in the estimator.  Let \(\alpha_\ell\) denote the effective
block-encoding scale of the level-\(\ell\) observable after known scalar prefactors have
been included.

\begin{enumerate}[label=(\roman*),leftmargin=2em]
\item Suppose amplitude estimation estimates \(\Delta_\ell\) to additive accuracy
\(\epsilon_\ell\), with
\[
\sum_{\ell=0}^L\epsilon_\ell\le \epsilon,
\]
and with level cost
\[
\mathcal C_\ell^{\rm AE}(\epsilon_\ell)
=
\widetilde O\!\left(\frac{c_\ell}{\epsilon_\ell}\right).
\]
Then the optimized total cost satisfies
\begin{equation}
\mathcal C_{\rm ML}^{\rm AE}
=
\widetilde O\!\left(
\frac1\epsilon
\left[\sum_{\ell=0}^L c_\ell^{1/2}\right]^2
\right).
\label{eq:mlmc-ae-allocation}
\end{equation}
In particular, for a block-encoding observable estimator,
\(
c_\ell=\tau_\ell\alpha_\ell,
\)
and therefore
\begin{equation}
\mathcal C_{\rm ML}^{\rm AE}
=
\widetilde O\!\left(
\frac1\epsilon
\left[\sum_{\ell=0}^L (\tau_\ell\alpha_\ell)^{1/2}\right]^2
\right).
\label{eq:mlmc-ae-alpha}
\end{equation}

\item Suppose direct sampling uses independent level estimators.  Let \(v_\ell\) be the
variance of one level-\(\ell\) sample, and let \(\tau_\ell\) be the cost of producing one such
sample.  Then the optimized cost to make the total sampling variance at most
\(\epsilon^2\) is
\begin{equation}
\mathcal C_{\rm ML}^{\rm samp}
=
\widetilde O\!\left(
\frac1{\epsilon^2}
\left[\sum_{\ell=0}^L (\tau_\ell v_\ell)^{1/2}\right]^2
\right).
\label{eq:mlmc-samp-allocation}
\end{equation}
If one level-\(\ell\) sample is bounded in magnitude by \(O(\alpha_\ell)\), then
\[
v_\ell\le C\alpha_\ell^2
\]
for a mesh-independent constant \(C\), and hence the worst-case sampling cost is
\begin{equation}
\mathcal C_{\rm ML}^{\rm samp}
=
\widetilde O\!\left(
\frac1{\epsilon^2}
\left[\sum_{\ell=0}^L \tau_\ell^{1/2}\alpha_\ell\right]^2
\right).
\label{eq:mlmc-samp-alpha}
\end{equation}
These allocation bounds use continuous counts; rounding up
adds at most $O(\sum_{\ell=0}^L\tau_\ell)$ to the cost.
\end{enumerate}
\end{lemma}

\begin{proof}
For amplitude estimation, we minimize
\[
\sum_{\ell=0}^L \frac{c_\ell}{\epsilon_\ell}, \quad \text{subject to} \; \sum_{\ell=0}^L\epsilon_\ell\le \epsilon. 
\]

The Lagrange multiplier equations give
\[
\epsilon_\ell
=
\epsilon
\frac{c_\ell^{1/2}}{\sum_{j=0}^L c_j^{1/2}}, \Longrightarrow \sum_{\ell=0}^L\frac{c_\ell}{\epsilon_\ell}
=
\frac1\epsilon
\left(\sum_{\ell=0}^L c_\ell^{1/2}\right)^2.
\]

For block-encoding observable estimation, the amplitude estimation approach gives
\(c_\ell=\tau_\ell\alpha_\ell\), which yields \eqref{eq:mlmc-ae-alpha}.

For direct sampling, let \(n_\ell\) be the number of independent level-\(\ell\) samples.
The total sampling variance and total cost are, respectively, 
\[
\sum_{\ell=0}^L \frac{v_\ell}{n_\ell},
\qquad
\sum_{\ell=0}^L n_\ell \tau_\ell.
\]
Minimizing the cost subject to $\sum_{\ell=0}^L \frac{v_\ell}{n_\ell}\le \epsilon^2$,
gives
\[
n_\ell\propto \sqrt{\frac{v_\ell}{\tau_\ell}} \Rightarrow \frac1{\epsilon^2}
\left(\sum_{\ell=0}^L (\tau_\ell v_\ell)^{1/2}\right)^2.
\]

If the sample is bounded by \(O(\alpha_\ell)\), then \(v_\ell\le C\alpha_\ell^2\), giving
\eqref{eq:mlmc-samp-alpha}.
\end{proof}

The rest of the paper constructs level operators with scale
$O(h_\ell^{2-\chi})$.  The two powers arise from the elliptic inverse on the
detail space, as shown by dual approximation in \Cref{subsec:J-oracle}.
They are not automatically replaced by a higher smooth-solution convergence
order when the polynomial degree is increased.

\section{Achieving multilevel reduction by two-level corrections }
\label{sec:two-level-estimators}

This section turns the multilevel identity into a quantum estimator.  We first
express each level difference through a corrected inverse, then use hierarchical
block elimination to expose its small scale before measurement.

\subsection{Level differences through discrete Green operators}
\label{subsec:D-level-differences}

Recall that $B_\ell=A_\ell^{-1}$ is the discrete Green operator.  Its lifted
coarse response and the correction supplied by refinement are
\begin{equation}\label{Cl-Dl}
 C_\ell =P_\ell B_{\ell-1}P_\ell^\dagger,
\qquad
D_\ell=B_\ell- C_\ell .
\end{equation}

\begin{lemma}[Level differences as one overlap or one expectation]
\label{lem:D-level-difference}
Assume \eqref{eq:galerkin-consistency}, \eqref{eq:g-compatible}, and
\eqref{eq:M-compatible}.
\begin{enumerate}[label=(\alph*),leftmargin=2em]
\item The linear level difference satisfies
\begin{equation}
\Delta_\ell^{\rm lin}
:=\bm g_\ell^\dagger B_\ell\bm b_\ell
-\bm g_{\ell-1}^\dagger B_{\ell-1}\bm b_{\ell-1}
=
\bm g_\ell^\dagger D_\ell\bm b_\ell.
\label{eq:D-linear-difference}
\end{equation}
Thus the level increment is a single overlap involving the residual-correction response.

\item The quadratic level difference satisfies
\begin{align}
\Delta_\ell^{\rm quad}
&:=
\bm b_\ell^\dagger B_\ell^\dagger M_\ell B_\ell\bm b_\ell
-
\bm b_{\ell-1}^\dagger B_{\ell-1}^\dagger M_{\ell-1}B_{\ell-1}\bm b_{\ell-1}
\notag\\
&=
\bm b_\ell^\dagger
\left(
D_\ell^\dagger M_\ell D_\ell
+D_\ell^\dagger M_\ell C_\ell 
+ C_\ell ^\dagger M_\ell D_\ell
\right)
\bm b_\ell.
\label{eq:D-quadratic-difference}
\end{align}
\end{enumerate}
\end{lemma}

\begin{proof}
The linear identity follows from
\[
\bm g_{\ell-1}^\dagger B_{\ell-1}\bm b_{\ell-1}
=
\bm g_\ell^\dagger P_\ell B_{\ell-1}P_\ell^\dagger\bm b_\ell
=
\bm g_\ell^\dagger C_\ell \bm b_\ell.
\]
Subtracting this from \(\bm g_\ell^\dagger B_\ell\bm b_\ell\) gives
\eqref{eq:D-linear-difference}.  For the quadratic identity, compatibility gives
\[
\bm b_{\ell-1}^\dagger B_{\ell-1}^\dagger M_{\ell-1}B_{\ell-1}\bm b_{\ell-1}
=
\bm b_\ell^\dagger C_\ell ^\dagger M_\ell C_\ell \bm b_\ell.
\]
Since \(B_\ell= C_\ell +D_\ell\), expanding
\(B_\ell^\dagger M_\ell B_\ell- C_\ell ^\dagger M_\ell C_\ell \) gives
\eqref{eq:D-quadratic-difference}.
\end{proof}

The identities place the cancellation inside a single overlap or expectation.
To obtain the multilevel gain, the block encoding of $D_\ell$ must also reflect
its small norm.  The hierarchical decomposition provides a way to do this.

\paragraph{The measured circuit.}
For a linear correction, let $U_D$ encode $D_\ell/\alpha_D$, and let
$U_b,U_g$ prepare the padded normalized data states.  A Hadamard test on
$(I\otimes U_g)^\dagger U_D(I\otimes U_b)$, with every block ancilla
initially zero, returns a sign $X\in\{-1,1\}$ with mean
$\operatorname{Re}\langle g_\ell|D_\ell|b_\ell\rangle/\alpha_D$.
Multiplying by $\|\bm g_\ell\|_2\|\bm b_\ell\|_2\alpha_D$ gives a
bounded estimator of the real output.  An imaginary-part Hadamard test handles
a complex linear output with a second estimate.

For a Hermitian quadratic readout, define
\[
 K_\ell=D_\ell^\dagger M_\ell C_\ell,\qquad
 Y_\ell=D_\ell^\dagger M_\ell D_\ell+K_\ell+K_\ell^\dagger.
\]
Product encodings use separate factor ancillas; the circuit for the third
term is the adjoint of the second.  The three terms have weights
$w_0=\alpha_D^2\alpha_M$ and
$w_1=w_2=\alpha_D\alpha_M\alpha_C$.
Prepare the three-label state
$\sum_{j=0}^2\sqrt{w_j/\alpha_{Y}}\ket{j}$, select the
corresponding product unitary, and unprepare the labels, where
\begin{equation}
 \alpha_{Y}=\alpha_M(\alpha_D^2+2\alpha_D\alpha_C).
\label{eq:quadratic-LCU-normalization}
\end{equation}
The resulting block encodes the Hermitian sum, even though $K_\ell$ alone
need not be Hermitian.  Its Hadamard test with $\ket{b_\ell}$ has real mean
$\Delta_\ell^{\rm quad}/(\|\bm b_\ell\|_2^2\alpha_{Y})$.
Thus one sample is bounded by the effective level normalization.  Amplitude
estimation uses the same coherent preparation and its inverse.

\begin{figure}[t]
\centering
\fbox{\begin{minipage}{0.91\linewidth}
\small
\textbf{Algorithm: multilevel estimation of a real output.}
\begin{enumerate}[label=\arabic*.,leftmargin=1.6em,itemsep=3pt]
\item \textit{Inputs.} Fix $L$, tolerance $\epsilon$, compatible load and
readout oracles, and encodings of $D_\ell$ and, for quadratic outputs,
$C_\ell$.  Include $Q_0$ as level zero.
\item \textit{Implementation precision.} Set
$\eta=c_{\rm imp}\epsilon/(L+1)$, where $c_{\rm imp}>0$
is a sufficiently small mesh-independent constant ensuring
total implementation error at most $\epsilon/2$. 
\item \textit{Level circuit.} Use the linear overlap above, or the
three-term LCU with normalization \eqref{eq:quadratic-LCU-normalization}.
Let $\alpha_\ell$ include the physical input-norm factors.
\item \textit{Estimate the level corrections.}
Choose one of the following methods:
\begin{itemize}
\item \textit{Amplitude estimation: fewer repetitions, longer coherent
circuits.} Allocate
\[
 \epsilon_\ell
 =\frac{\epsilon}{2}
 \frac{\sqrt{\tau_\ell\alpha_\ell}}
 {\sum_j\sqrt{\tau_j\alpha_j}}.
\]
Estimate each normalized sign mean to accuracy
$\epsilon_\ell/\alpha_\ell$ (equivalently, the test probability
to $\epsilon_\ell/(2\alpha_\ell)$).
Median amplification reduces the failure probability at each
level to at most $1/[3(L+1)]$.

\item \textit{Direct sampling: shorter coherent circuits, more repetitions.}
Take independent samples with counts
$n_\ell\propto\alpha_\ell/\sqrt{\tau_\ell}$, scaled so that
\[
 \sum_\ell\frac{\alpha_\ell^2}{n_\ell}\le\frac{\epsilon^2}{4},
\]
and round each count up.
\end{itemize}
\item \textit{Output.} Restore the known scalar factors and sum the
estimated contributions.  
\end{enumerate}
\end{minipage}}
\caption{The estimator uses the level differences coherently before
measurement.  The cost and error guarantees are given in
\Cref{thm:overall-complexity}.}
\label{fig:estimator-algorithm}
\end{figure}

\paragraph{Depth and repetitions.}
For a real contribution $\Delta_\ell$ with known effective normalization
$\alpha_\ell$, the sign test has mean $\Delta_\ell/\alpha_\ell$ and success probability
$\tfrac12(1+\Delta_\ell/\alpha_\ell)$.  Standard phase-estimation-based
amplitude estimation therefore uses
\[
 q_\ell=O\!\left(\max\{1,\alpha_\ell/\epsilon_\ell\}\right)
\]
controlled test-circuit and inverse calls for an additive error
$\epsilon_\ell$, with a constant success probability.  If one such call has
gate count $\tau_\ell$ and depth $  d_\ell$, the resulting circuit has gate
count $O(q_\ell\tau_\ell)$ and coherent depth $O(q_\ell   d_\ell)$.
It uses $O(\log q_\ell)$ additional phase-estimation qubits.  Taking the
median of $O(\log((L+1)/\delta))$ independent estimates per level bounds
the total failure probability by $\delta$.
Thus the total gate count and maximum coherent depth are, respectively,
\[
 O\!\left(\log\frac{L+1}{\delta}
             \sum_{\ell=0}^Lq_\ell\tau_\ell\right),\qquad
 O\!\left(\max_{0\le\ell\le L}q_\ell   d_\ell\right),
\]
and the number of complete circuit repetitions is
$O((L+1)\log((L+1)/\delta))$.
Direct sampling instead uses one sign-test call per repetition, with depth
$O(  d_\ell)$, with only precision-logarithmic dependence on accuracy,
and combines the outcomes using the allocation in \Cref{sec:mlmc}.  Consequently the
$\epsilon^{-1}$ amplitude-estimation rate describes both total query work
and its long coherent circuits, whereas the $\epsilon^{-2}$ sampling rate
is achieved through many short circuits.

\subsection{Hierarchical splitting and the Schur complement}
\label{subsec:hierarchical-schur}

All matrices in this subsection use the reduced coordinates of
\Cref{subsec:response-coordinates}.  For reduced vectors $\bm x_\ell$ and $\bm y_\ell$ representing
$u_\ell$ and $v_\ell$, respectively, we have
\begin{equation}
 u_\ell(x)=h_\ell^{-d/2}\Phi_\ell(x)^\dagger\bm x_\ell,
 \qquad a(u_\ell,v_\ell)=\bm y_\ell^\dagger A_\ell\bm x_\ell.
\label{eq:function-from-reduced-vector}
\end{equation}
The matrix $P_\ell$ represents the inclusion $V_{\ell-1}\subset V_\ell$.

Let \(m_\ell:=N_\ell-N_{\ell-1}\), and choose a matrix
\(W_\ell\in\mathbb R^{N_\ell\times m_\ell}\) such that
\begin{equation}
\mathbb R^{N_\ell}
=
\range(P_\ell)\oplus\range(W_\ell),
\label{eq:complement-splitting}
\end{equation} and $
T_\ell:=[P_\ell\; W_\ell]\in\mathbb R^{N_\ell\times N_\ell}$
is invertible.  The columns of $W_\ell$ specify an algebraic detail space; they need not yet
be energy-orthogonal to the coarse space.  Every vector has the unique form
$\bm x_\ell=P_\ell\bm y_{\ell-1}+W_\ell\bm z$.
Such decompositions are standard in subspace correction
\cite{Xu1992,Oswald1994}; the algebra below is ordinary Schur elimination
\cite{Hackbusch1985,TrottenbergOosterleeSchuller2001}.

Let us consider the stiffness matrix in these new coordinates,
\begin{equation}
\hat{A}_\ell
:=
T_\ell^\dagger A_\ell T_\ell
=
\begin{bmatrix}
A_{\ell-1} & E_\ell\\
E_\ell^\dagger & H_\ell
\end{bmatrix},
\label{eq:T-coordinate-stiffness}
\end{equation}
where
\begin{equation}
E_\ell:=P_\ell^\dagger A_\ell W_\ell,
\qquad
H_\ell:=W_\ell^\dagger A_\ell W_\ell .
\label{eq:E-H}
\end{equation}
The matrix \(E_\ell\) is the energy coupling between the embedded coarse space
\(\range(P_\ell)\) and the chosen complement \(\range(W_\ell)\).

The Ritz projection onto the embedded coarse space is represented in reduced coordinates
by
\begin{equation}
\Pi_{\ell-1}^{a}
:=
P_\ell A_{\ell-1}^{-1}P_\ell^\dagger A_\ell .
\label{eq:ritz-projection-coordinate}
\end{equation}
More precisely, for the PDE weak form, if \(\bm v\) represents \(v_\ell\), then
\(\Pi_{\ell-1}^{a}\bm v\) represents its coarse Ritz projection
\(\mathcal R_{\ell-1}^{a}v_\ell\in V_{\ell-1}\), characterized by
\[
a(v_\ell-\mathcal R_{\ell-1}^{a}v_\ell,q_{\ell-1})=0
\qquad\forall q_{\ell-1}\in V_{\ell-1}.
\]
The corresponding Ritz complement is \(I-\Pi_{\ell-1}^{a}\).  Applying this complement to
the chosen complement basis \(W_\ell\) gives the Ritz-complement coordinate matrix
\begin{equation}
J_\ell
:=
(I-\Pi_{\ell-1}^{a})W_\ell
=
W_\ell-P_\ell A_{\ell-1}^{-1}E_\ell .
\label{eq:J-as-ritz-complement}
\end{equation}

The associated Ritz-complement basis functions are
\begin{equation}
\Psi_\ell^{J}(x)^\dagger
:=
h_\ell^{-d/2}\Phi_\ell(x)^\dagger J_\ell .
\label{eq:J-function-basis}
\end{equation}
Thus, if \(\bm z\in\mathbb R^{m_\ell}\), the finite element function represented by
the reduced-coordinate vector \(J_\ell\bm z\) is
\[
u_\ell(x)=\Psi_\ell^J(x)^\dagger\bm z .
\]

Finally, we turn to the Schur complement
\begin{equation}
\Sigma_\ell
:=
H_\ell-E_\ell^\dagger A_{\ell-1}^{-1}E_\ell .
\label{eq:Sigma}
\end{equation}
The interpretation of \(J_\ell\) and \(\Sigma_\ell\) is simple: \(J_\ell\) represents the
chosen complement after removal of its coarse Ritz projection, while \(\Sigma_\ell\) is the
energy matrix on this Ritz-complement coordinate space.  

\begin{lemma}[Ritz decomposition and corrected Green operator]
\label{lem:ritz-decomposition}
\(\Pi_{\ell-1}^{a}\) is the \(A_\ell\)-orthogonal projector onto
\(\range(P_\ell)\), \(J_\ell\) has full column rank, and
\[
 \range(J_\ell)
 =\{\bm v\in\mathbb R^{N_\ell}:P_\ell^\dagger A_\ell\bm v=0\}.
\]
In particular,
\begin{equation}
P_\ell^\dagger A_\ell J_\ell=0,
\label{eq:J-ritz-orthogonality}
\end{equation}
and
\begin{equation}
\Sigma_\ell=J_\ell^\dagger A_\ell J_\ell .
\label{eq:Sigma-JAJ}
\end{equation}
The matrix \(\Sigma_\ell\) is symmetric positive definite, and
\begin{equation}
A_\ell^{-1}
=P_\ell A_{\ell-1}^{-1}P_\ell^\dagger
{}+J_\ell\Sigma_\ell^{-1}J_\ell^\dagger.
\label{eq:ritz-inverse-decomposition}
\end{equation}
Consequently,
\begin{equation}
D_\ell
=A_\ell^{-1}-P_\ell A_{\ell-1}^{-1}P_\ell^\dagger
=J_\ell\Sigma_\ell^{-1}J_\ell^\dagger
\succeq0.
\label{eq:D-schur}
\end{equation}
If \(u_\ell(x)=\Psi_\ell^J(x)^\dagger\bm z\), then its energy is
\[
\bm z^\dagger\Sigma_\ell\bm z
=a(u_\ell,u_\ell).
\]

\end{lemma}

The same statement holds for complex coordinates with Hermitian positive-definite
matrices, as used in the spectral example below.

\begin{proof}
The formula \eqref{eq:ritz-projection-coordinate} and Galerkin consistency show directly
that \((\Pi_{\ell-1}^{a})^2=\Pi_{\ell-1}^{a}\),
\(\range(\Pi_{\ell-1}^{a})=\range(P_\ell)\), and
\(P_\ell^\dagger A_\ell(I-\Pi_{\ell-1}^{a})=0\).  Hence
\eqref{eq:J-ritz-orthogonality} holds.  If \(J_\ell\bm z=0\), then
\(W_\ell\bm z\in\range(P_\ell)\); the direct-sum property of
\([P_\ell\;W_\ell]\) therefore implies \(\bm z=0\).  A dimension count now gives the
asserted range identity.

Set
\[
L_\ell^{\rm R}:=
\begin{bmatrix}I&-A_{\ell-1}^{-1}E_\ell\\0&I\end{bmatrix}.
\]
Then \([P_\ell\;J_\ell]=T_\ell L_\ell^{\rm R}\), and direct multiplication gives
\[
(L_\ell^{\rm R})^\dagger
\begin{bmatrix}A_{\ell-1}&E_\ell\\E_\ell^\dagger&H_\ell\end{bmatrix}
L_\ell^{\rm R}
=
\begin{bmatrix}A_{\ell-1}&0\\0&\Sigma_\ell\end{bmatrix}.
\]
This proves \eqref{eq:Sigma-JAJ}.  Since \(A_\ell\) is positive definite and
\(J_\ell\) has full column rank, \(\Sigma_\ell\) is positive definite.  Inverting the
diagonalized congruence proves \eqref{eq:ritz-inverse-decomposition}, and subtracting the
coarse term gives \eqref{eq:D-schur}.
\end{proof}

\subsection{Analytic scaling and quantum access}
\label{subsec:J-oracle}

The factorization
\[
D_\ell=J_\ell\Sigma_\ell^{-1}J_\ell^\dagger,
\qquad \Sigma_\ell=J_\ell^\dagger A_\ell J_\ell,
\]
connects a familiar multigrid object with the quantum readout problem.  The Ritz
projection removes the part resolved on the coarse grid.  Its remaining response is
small in the physical $L^2$ norm. However, to reduce measurement cost, the circuit must preserve
this small scale.  We first explain the approximation estimate behind that statement,
then specify the coordinates needed to implement it.

\begin{lemma}[Size of the Ritz-detail response]
\label{lem:J-L2-scale}
Under the approximation property \eqref{eq:dual-approximation},  a function
$u_\ell=\Psi_\ell^J(\cdot)^\dagger\bm z$ satisfies
\begin{equation}
 \|u_\ell\|_{L^2}\le C h_\ell\|u_\ell\|_a,
 \qquad
 \|J_\ell\bm z\|_2
 \le C h_\ell\|\Sigma_\ell^{1/2}\bm z\|_2.
\label{eq:J-L2-scale}
\end{equation}
Consequently, independently of the choice of detail basis,
\begin{equation}
 \|D_\ell\|_2
 =\|J_\ell\Sigma_\ell^{-1/2}\|_2^2
 =O(h_\ell^{2}).
\label{eq:D-analytic-scale}
\end{equation}
If the energy Gram matrices also satisfy $\Sigma_\ell\preceq C_0I$, then
\begin{equation}
 \|J_\ell\|_2=O(h_\ell).
\label{eq:J-reduced-norm}
\end{equation}
\end{lemma}

\begin{proof}
Choose $g=u_\ell$ in \Cref{eq:dual-approximation}.  Ritz orthogonality gives
\[
 \|u_\ell\|_{L^2}^2
 =a(u_\ell,\phi_{u_\ell}-q)
 \le \|u_\ell\|_a\|\phi_{u_\ell}-q\|_a
 \qquad(q\in V_{\ell-1}).
\]
Take the infimum and use $h_{\ell-1}=2h_\ell$.  The reduced-coordinate norm
equivalence \eqref{eq:nodal-L2-equivalence} and
$\|u_\ell\|_a^2=\bm z^\dagger\Sigma_\ell\bm z$ prove
\eqref{eq:J-L2-scale}.  Substitute $\bm z=\Sigma_\ell^{-1/2}\bm y$ and use
\eqref{eq:D-schur} to obtain \eqref{eq:D-analytic-scale}.
\end{proof}

This lemma highlights the importance of energy stability of the detail
coordinates. Related stability and norm-equivalence estimates play a
central role in classical multilevel preconditioning and subspace
correction \cite{BramblePasciakXu1990,Xu1992}.
When $c_0I\preceq\Sigma_\ell\preceq C_0I$ with level-independent
constants, the lemma gives $\|J_\ell\|_2=O(h_\ell)$, while
$\|\Sigma_\ell^{-1}\|_2=O(1)$. Thus the factorization of $D_\ell$
combines a small detail map with a uniformly bounded inverse.
The next proposition establishes these energy bounds for standard
nodal details.

\begin{proposition}[Energy stability of nodal Ritz details]
\label{prop:nodal-ritz-stability}
Consider the conforming $P_1$ hierarchy above, or the corresponding affine
$Q_1$ hierarchy.  Suppose each refinement step is drawn from a fixed finite
set of shape-regular reference refinements and $h_{\ell-1}=2h_\ell$.
Let $S_\ell^{\rm new}$ inject the newly introduced interior nodal coordinates,
set $W_\ell=h_\ell S_\ell^{\rm new}$, and form $J_\ell$ and $\Sigma_\ell$ by
\eqref{eq:J-as-ritz-complement}.  Then there are level-independent constants
$0<c_0\le C_0<\infty$ such that
\begin{equation}
 c_0\|\bm z\|_2^2\le\bm z^\dagger\Sigma_\ell\bm z
 \le C_0\|\bm z\|_2^2.
\label{eq:stable-ritz-complement}
\end{equation}
The constants may depend on dimension, shape regularity, the refinement
patterns, and the ellipticity bounds. 
\end{proposition}
\begin{proof}
Write $h=h_\ell$, $H=h_{\ell-1}$, and let $w\in V_\ell$ be represented in
reduced coordinates by $W_\ell\bm z$.  Nodal $L^2$ norm equivalence and the
inverse inequality give
\[
 c_Mh^2\|\bm z\|_2^2\le\|w\|_{L^2}^2
 \le C_Mh^2\|\bm z\|_2^2,\qquad
 \|\nabla w\|_{L^2}\le C_{\rm inv}h^{-1}\|w\|_{L^2}.
\]
Let $I_H:V_\ell\to V_{\ell-1}$ be interpolation at the coarse nodes.
On this fixed-ratio pair of spaces, scaling the finitely many reference
refinements and using norm equivalence modulo constants yields
\[
 \|v-I_Hv\|_{L^2}\le C_I H\|\nabla v\|_{L^2}
 \qquad(v\in V_\ell).
\]
This is a statement on the fine finite element space, not a claim that
pointwise interpolation is bounded on all of $H^1$.
Since $I_Hw=0$, for every $q\in V_{\ell-1}$ the function $v=w+q$ satisfies
$v-I_Hv=w$.  The variational characterization of the Schur complement gives
\[
 \bm z^\dagger\Sigma_\ell\bm z
 =\min_{q\in V_{\ell-1}}a(w+q,w+q)
 \ge\frac{a_{\min}}{C_I^2H^2}\|w\|_{L^2}^2
 \ge\frac{a_{\min}c_M}{4C_I^2}\|\bm z\|_2^2.
\]
Taking $q=0$ gives the upper bound with
$C_0=a_{\max}C_{\rm inv}^2C_M$.
\end{proof}

For this standard two-level nodal splitting, stability is therefore a
classical finite element property.  It does not assert uniform stability
of an uncorrected hierarchical basis assembled over all levels.   Consequently, the additional algorithmic requirement is access to the
Ritz-corrected map at its natural scale.

\begin{assumption}[Quantum access to Ritz coordinates]
\label{ass:J-oracle}
Choose Ritz coordinates satisfying \eqref{eq:stable-ritz-complement}.
For every requested accuracy $\eta$, $J_\ell$ has a block encoding with normalization, error, and gate counts, respectively
\[
 \alpha_{J_\ell}=O(h_\ell),\qquad
 \epsilon_{J_\ell}\le\eta h_\ell,\qquad
 T_{J_\ell}(\eta)=\operatorname{polylog}(N_\ell,\eta^{-1}).
\]
The circuit can be selected coherently by the level register, with its
coefficients computed reversibly on demand, without an $O(N_\ell)$
preprocessing table.
\end{assumption}

\begin{theorem}[Block encoding of the corrected Green operator]
\label{thm:sigma-from-J}
Assume the energy bounds \eqref{eq:stable-ritz-complement} and
\Cref{ass:J-oracle}.  For every $0<\eta\le\eta_0$, suppose $A_\ell$ has a
block encoding with normalization $O(h_\ell^{-2})$, operator error
$O(\eta h_\ell^{-2})$, and cost $T_{A_\ell}(\eta)$.  Then $\Sigma_\ell$ and
$\Sigma_\ell^{-1}$ have $O(1)$-normalized block encodings with error $O(\eta)$.
Moreover, $D_\ell$ has a block encoding with normalization $O(h_\ell^2)$, operator
error $O(\eta h_\ell^2)$, and cost
\[
 \widetilde O\!\left(T_{J_\ell}(\eta)+T_{A_\ell}(\eta)\right),
\]
where $\eta_0>0$ is mesh independent and the tilde suppresses precision logarithms.
More explicitly, if $\alpha_J\le C_J(d)h_\ell$ and
$\alpha_A\le C_A(d)h_\ell^{-2}$, the inverse query multiplier is
$O(C_J(d)^2C_A(d)/c_0)$, up to precision logarithms, and
$\alpha_D=O(C_J(d)^2h_\ell^2/c_0)$.
The constants are not claimed uniform in $d$ or coefficient contrast.
\end{theorem}

\begin{proof}
The product rule gives
$\alpha_{\Sigma_\ell}\le\alpha_{J_\ell}^2\alpha_{A_\ell}=O(1)$.
Regard the implemented detail and stiffness blocks as exact matrices
$\widetilde J_\ell$ and $\widetilde A_\ell$ satisfying the stated error
bounds.  Products give an exact encoding of
$X_\ell=\widetilde J_\ell^\dagger\widetilde A_\ell\widetilde J_\ell$ and
\[
 \|X_\ell-\Sigma_\ell\|\le C\eta,
\]
because every input error is proportional to its mesh scale.
The equal-weight LCU of the product circuit and its adjoint encodes
$\Sigma_\ell^{\rm imp}=(X_\ell+X_\ell^\dagger)/2$ with the same $O(1)$ normalization.
On padded registers, compress to the valid detail indices,
extend by $c_0I$ on the padding, and project back after inversion.
For $C\eta\le c_0/2$, its spectrum is contained in
$[c_0/2,C_0+c_0/2]$.  
Apply \Cref{lem:inverse-block-encoding} to
$\Sigma_\ell^{\rm imp}$ with absolute approximation tolerance
$O(\eta)$. The resolvent identity gives
\[
 \|(\Sigma_\ell^{\rm imp})^{-1}-\Sigma_\ell^{-1}\|
 \le \frac{2}{c_0^2}
    \|\Sigma_\ell^{\rm imp}-\Sigma_\ell\|
 =O(\eta).
\]
Thus the inverse block has $O(1)$ normalization and error
$O(\eta)$ relative to $\Sigma_\ell^{-1}$.
Multiplication by the two detail blocks gives normalization $O(h_\ell^2)$
and total absolute error $O(\eta h_\ell^2)$.  QSVT phase and other elementary
rotations can be synthesized to a total operator error of order $\eta$
inside the normalized inverse circuit, contributing the same final error
scale and only precision logarithms to the gate count.
\end{proof}

Quadratic outputs also involve the lifted coarse response $C_\ell$.  Its efficient
implementation follows from the same corrections, provided transfer between distant
levels can be encoded without accumulating normalization losses.

\begin{lemma}[Lifted coarse response]
\label{lem:coarse-response-block-encoding}
Let $P_{j\to\ell}:=P_\ell\cdots P_{j+1}$.  Suppose $B_0$ has an
$O(1)$-normalized block encoding.  Suppose each cumulative transfer $P_{j\to\ell}$
has normalization at most a mesh-independent $C_P$, polylogarithmic cost and
precision dependence, and coherent selection over $j$.  If
$\alpha_{D_j}=O(h_j^2)$ with the same access properties, then $C_\ell$ has an
$O(1)$-normalized block encoding of polylogarithmic cost.
\end{lemma}

\begin{proof}
The inverse telescoping identity gives
\[
 C_\ell=P_{0\to\ell}B_0P_{0\to\ell}^\dagger
 +\sum_{j=1}^{\ell-1}P_{j\to\ell}D_jP_{j\to\ell}^\dagger.
\]
Its LCU normalization is at most
$C_P(d)^2(\alpha_{B_0}+\sum_{j<\ell}\alpha_{D_j})=O(1)$ for fixed $d$.
Thus the square of the cumulative-transfer normalization enters quadratic
costs. 

\end{proof}

\subsection{Examples for efficient Ritz-complement coordinates}
\label{subsec:J-examples}
We now give 
 two examples to illustrate  ways to fulfill the access assumption for the detail scale.
\begin{example}[Exact midpoint details for one-dimensional \(P_1\) elements]
\label{ex:midpoint-ritz}
Let \(\Omega=(0,1)\), let \(\mathcal T_\ell\) be the uniform dyadic mesh obtained by
bisecting every \(K_j\in\mathcal T_{\ell-1}\), and write \(h=h_\ell\), so
\(|K_j|=2h\).  Assume that
the scalar diffusion coefficient is constant, with value \(a_j\), on each parent
element.  This holds simultaneously on every level, for example, if the coefficient is
piecewise constant on the coarsest mesh and all coefficient interfaces are preserved.

Let \(S_\ell^{\rm mid}\bm e_j\) be the fine-grid coordinate vector of the new midpoint
of \(K_j\), and set
\begin{equation}
 W_\ell=h S_\ell^{\rm mid}.
\label{eq:midpoint-W}
\end{equation}
The factor $h$ converts an energy-scaled midpoint hat into the reduced
coordinates: the physical function represented by its $j$-th column is
\[
 h^{-1/2}\Phi_\ell^\dagger W_\ell\bm e_j
 =h^{1/2}\varphi_{m_j}^{(\ell)},
\]
where \(m_j\) is the midpoint of \(K_j\).  If \(q_{\ell-1}\in V_{\ell-1}\), then
\[
 a(h^{1/2}\varphi_{m_j}^{(\ell)},q_{\ell-1})
 =h^{1/2}a_j q_{\ell-1}'|_{K_j}
   \int_{K_j}(\varphi_{m_j}^{(\ell)})'\,dx
 =0.
\]
The midpoint details on distinct parents have disjoint interiors, and a direct element
calculation gives
\[
 a(h^{1/2}\varphi_{m_j}^{(\ell)},
   h^{1/2}\varphi_{m_k}^{(\ell)})
 =2a_j\delta_{jk},
\]
implying the energy orthogonality \cite{Yserentant1986,Oswald1994}.

Consequently, 
\begin{equation}
 P_\ell^\dagger A_\ell W_\ell=0,
 \qquad
 J_\ell=W_\ell,
 \qquad
 \Sigma_\ell=2\diag(a_j),
\label{eq:midpoint-J-Sigma}
\end{equation}
and
\[
 2a_{\min}I\preceq\Sigma_\ell\preceq2a_{\max}I,
 \qquad
 \|J_\ell\|_2=h.
\]
Thus \Cref{ass:J-oracle} holds exactly: the midpoint injection
has an \(h\)-normalized sparse block encoding.  In fact the corrected inverse itself is
\begin{equation}
 D_\ell
 =h^2S_\ell^{\rm mid}\diag\!\left(\frac{1}{2a_j}\right)
   (S_\ell^{\rm mid})^\dagger.
\label{eq:midpoint-D}
\end{equation}
It has a direct block encoding with normalization \(h^2/(2a_{\min})\), using only a
midpoint predicate and reversible access to $a_j$.  We take the coarsest
coefficient partition to have fixed size; more general coefficient data require
their own efficient value oracle.

The cumulative transfers needed for quadratic outputs are also explicit.
For $r=h_j/h_\ell=2^{\ell-j}$, $P_{j\to\ell}$ is ordinary nodal interpolation
multiplied by $r^{-1/2}$.  It has at most two nonzeros per row, $O(r)$ per
column, and entries bounded by $r^{-1/2}$.  Consequently,
\Cref{lem:sparse-block-encoding} gives normalization
$O(\sqrt{2r}\,r^{-1/2})=O(1)$.  The support indices and linear interpolation
weights follow from integer arithmetic on the node and level indices, so this
construction has polylogarithmic cost even for distant levels.  It verifies
\Cref{lem:coarse-response-block-encoding} without multiplying the normalizations
of successive transfers.

\end{example}
\begin{example}[Nested Fourier modes]
\label{ex:sine-qft-ritz}
Consider Fourier--Galerkin discretizations of $-\Delta+\mu I$,
with fixed $\mu>0$, on the periodic domain $(0,1)^d$. Set
$n_\ell=2^{r_\ell}$, $r_\ell=r_0+\ell$, $r_0\ge1$, and
$h_\ell=n_\ell^{-1}$. The retained modes and the newly added band are
\[
 \mathcal K_\ell=\{-n_\ell/2,\ldots,n_\ell/2-1\}^d,
 \qquad
 \mathcal B_\ell=\mathcal K_\ell\setminus\mathcal K_{\ell-1}.
\]
Using the unitary discrete Fourier transform \(F_\ell\), write
\[
 A_\ell=F_\ell^\dagger
 \operatorname{diag}_{\bm k\in\mathcal K_\ell}
 (\lambda_{\bm k})F_\ell,
 \qquad
 \lambda_{\bm k}=\mu+4\pi^2|\bm k|_2^2.
\]
Coarse functions are embedded by retaining their Fourier coefficients.
Since each retained mode has the same eigenvalue on every level,
the coarse contribution cancels exactly, giving
\begin{equation}
 D_\ell
 =F_\ell^\dagger
 \operatorname{diag}_{\bm k\in\mathcal K_\ell}
 \left(\frac{\mathbf 1_{\mathcal B_\ell}(\bm k)}
 {\lambda_{\bm k}}\right)F_\ell.
\label{eq:spectral-band-D}
\end{equation}

On the new band, $\lambda_{\bm k}\asymp h_\ell^{-2}$.
Choosing the new Fourier modes, scaled by $h_\ell$, as detail
columns therefore gives
\[
 \|J_\ell\|_2=h_\ell,\qquad
 \Sigma_\ell=h_\ell^2
 \operatorname{diag}_{\bm k\in\mathcal B_\ell}(\lambda_{\bm k}),
 \qquad
 cI\preceq\Sigma_\ell\preceq CI,
\]
with level-independent constants. Equation
\eqref{eq:spectral-band-D} gives a direct block encoding with
normalization $O(h_\ell^2)$ using the QFT and a diagonal multiplier.
Cumulative transfers retain the coarse mode labels and have
normalization one. 

The same construction applies to sine modes for Dirichlet problems
and cosine modes for shifted Neumann problems
\cite{CaoEtAl2013,HuangAntonioliBarbaresco2026}.
These constructions use nested spectral Galerkin spaces with
inherited eigenvalues.
\end{example}

\paragraph{General finite element constructions.}
The preceding examples exploit exact energy orthogonality, for which
$E_\ell=0$ and $J_\ell=W_\ell$. Classical hierarchical-basis and
multilevel methods also provide energy-stable nonorthogonal
decompositions
\cite{DahmenKunoth1992,Oswald1994,VassilevskiWang1997,VassilevskiWang1998}.
This motivates exploring detail coordinates that permit efficient
access to the full Ritz correction
\[
 J_\ell=W_\ell-P_\ell A_{\ell-1}^{-1}E_\ell,
 \qquad E_\ell=P_\ell^\dagger A_\ell W_\ell.
\]
We give a sufficient normalization criterion below and leave its
realization for general finite element hierarchies to future work.

Suppose $W_\ell$ and $P_\ell$ admit block encodings with
normalizations $O(h_\ell)$ and $O(1)$, respectively. To allow for
normalization losses in the coarse solve, consider an inverse
encoding with
\[
 \alpha_{A_{\ell-1}^{-1}}=O(h_\ell^{-2}).
\]
 If the coupling has the representation
\[
 E_\ell=h_\ell^q\widehat E_\ell,\qquad q\ge3,
\]
with an $O(1)$-normalized encoding of $\widehat E_\ell$, the product
and linear-combination rules yield
\[
 \alpha_{J_\ell}
 \le \alpha_{W_\ell}
 +h_\ell^q\alpha_{P_\ell}
  \alpha_{A_{\ell-1}^{-1}}\alpha_{\widehat E_\ell}
 =O(h_\ell+h_\ell^{q-2})
 =O(h_\ell).
\]
Thus both terms in the exact Ritz map retain the required scale.
Together with the energy bounds
\eqref{eq:stable-ritz-complement}, this gives the
$O(h_\ell^2)$ normalization of the corrected Green operator through
\Cref{thm:sigma-from-J}.
The corresponding precision and implementation costs must
be established separately.

The condition $q\ge3$ applies to the construction that encodes
$A_{\ell-1}^{-1}$ and $E_\ell$ separately. A broader possibility
is to construct the combined map directly, with
\[
 \alpha_{A_{\ell-1}^{-1}E_\ell}=O(h_\ell),
\]
even when the coupling itself does not satisfy the preceding decay
condition. Only the action of the coarse inverse on the range of
$E_\ell$ is needed. 

Structured finite element hierarchies provide possible settings for
such a construction. For constant-coefficient problems on uniform periodic Cartesian
$Q_1$ meshes, the harmonic decomposition used in local Fourier
analysis reduces the coarse--fine interaction to frequency
blocks of bounded size for fixed dimension 
\cite{TrottenbergOosterleeSchuller2001}. Evaluating the inverse action
and coupling jointly within these blocks offers a route to preserving
their combined scale. In one dimension, explicit discrete Green
functions provide another possibility when the coefficient-dependent
resistance sums can be evaluated efficiently.

These observations motivate further constructions based on the
combined coarse correction. Suitable first-order or preconditioned
representations may also reduce the inversion cost, but the required
normalization and implementation cost must be established together.
Developing and verifying such constructions, including their accuracy
and polylogarithmic access bounds, is left to future work.

\subsection{The overall complexity of the Schur-complement estimator}
\label{subsec:schur-complement-complexity}

We now assemble the ingredients.  The detail circuit supplies the small
correction $D_\ell$; the readout specifies how strongly that correction is
weighted; and multilevel allocation determines how accurately each contribution
must be measured.  This separates the cost of one circuit call from the number
of calls needed for the output.

Write $\tau_\ell$ for the cost of a controlled level circuit and its inverse,
including the load and readout preparations of \Cref{subsec:rhs-preparation},
reflections used in amplitude estimation, observable block encodings, and their
precision requirements.  For quadratic outputs it includes $C_\ell$ as well.
The fixed coarsest problem has an \(O(1)\) inverse normalization and only precision-logarithmic implementation cost per call.

For a linear output, \eqref{eq:D-linear-difference} and
$\alpha_{D_\ell}=O(h_\ell^2)$ give effective scale
\begin{equation}
 \alpha_\ell^{\rm lin}=O(h_\ell^{2-\chi}).
\label{eq:linear-level-scale-schur}
\end{equation}
For a quadratic output, \eqref{eq:D-quadratic-difference} contains
$D_\ell^\dagger M_\ell D_\ell$ with scale $O(h_\ell^{4-\chi})$ and two cross
terms with scale $O(h_\ell^{2-\chi})$, since $\alpha_{C_\ell}=O(1)$.  Therefore
\begin{equation}
 \alpha_\ell^{\rm quad}=O(h_\ell^{2-\chi}).
\label{eq:quadratic-level-scale-schur}
\end{equation}
The $O(1)$ bounds on load norms are included in both statements.

\begin{theorem}[Multilevel estimation of discrete elliptic outputs]
\label{thm:overall-complexity}
Consider a nested Galerkin hierarchy with compatible loads and readouts
\eqref{eq:galerkin-consistency}, \eqref{eq:g-compatible}, and \eqref{eq:M-compatible}, uniformly bounded
load norms, and admissible readout order $0\le\chi\le2$.  Suppose the load and
linear-readout states and their scalar norms are available under the compressed
data model of \Cref{subsec:rhs-preparation}.  Suppose $D_\ell$ has a block
encoding with normalization $O(h_\ell^2)$ and operator error
$O(\eta h_\ell^2)$ at polylogarithmic cost.  This is supplied by
\Cref{thm:sigma-from-J} or by the direct constructions in
\Cref{ex:midpoint-ritz,ex:sine-qft-ritz}.
For quadratic outputs, also assume the cumulative-transfer access of
\Cref{lem:coarse-response-block-encoding}.  All preparations, input norms, and
readout encodings are computable to precision proportional to their stated
scales.  Assume that every complete level call has cost
$\tau_\ell(\eta)=\operatorname{polylog}(N_\ell,\eta^{-1})$, including prefix-weight
evaluation, norm arithmetic, and readout access.

Then the finest discrete output $Q_L$ can be estimated to additive error
$\epsilon$ with probability at least $2/3$ by amplitude estimation at cost
\begin{equation}
 \mathcal C_{\rm Schur}^{\rm AE}
 =\widetilde O\!\left(
 \epsilon^{-1}\left[
 \sum_{\ell=0}^L(\tau_\ell h_\ell^{2-\chi})^{1/2}
 \right]^2\right)
 =\widetilde O(\epsilon^{-1}).
\label{eq:schur-AE-complexity}
\end{equation}
Independent direct sampling achieves root-mean-square error at most $\epsilon$
at cost
\begin{equation}
 \mathcal C_{\rm Schur}^{\rm samp}
 =\widetilde O\!\left(
 \epsilon^{-2}\left[
 \sum_{\ell=0}^L \tau_\ell^{1/2}h_\ell^{2-\chi}
 \right]^2\right)
 =\widetilde O(\epsilon^{-2}).
\label{eq:schur-sampling-complexity}
\end{equation}
These bounds apply to gates and to calls to the specified data oracles, with
their corresponding per-call costs.  Register and arithmetic workspace require
$\operatorname{polylog}(N_L,\epsilon^{-1})$ qubits.  The suppressed factors include powers of $L+1$, particularly
at $\chi=2$. Since $L=\log_2(h_0/h_L)$, these contribute
only polylogarithmic mesh dependence. When the mesh is chosen
to meet an algebraic output-bias bound,
$L=O(1+\log\epsilon^{-1})$, so these factors are also
polylogarithmic in the inverse tolerance.
\end{theorem}

\begin{proof}
The level identities and block-encoding algebra give
\eqref{eq:linear-level-scale-schur}--\eqref{eq:quadratic-level-scale-schur}.
Using relative scale accuracy $\eta$ for the matrix factors and scalar norms,
and state-vector accuracy $O(\eta)$ for the preparations, changes a level output
by at most $C\eta h_\ell^{2-\chi}$.  For example, an error
$O(\eta)$ in $C_\ell$ is multiplied by a detail factor of scale $h_\ell^2$
in every quadratic cross term.  Hence $\eta=c_{\rm imp}\epsilon/(L+1)$, with a sufficiently
small mesh-independent constant $c_{\rm imp}>0$, bounds the
total implementation bias by $\epsilon/2$.
It adds only precision logarithms to $\tau_\ell$.

Apply \Cref{lem:mlmc-allocation} to the remaining error budget.  For amplitude
estimation, take the median of $O(\log(L+1))$ independent constant-success
estimates per level to make its failure probability at most $1/[3(L+1)]$,
then use a union bound.  For sampling, allocate the variance budget
after the deterministic bias has been reserved.  Rounding sample and query
counts up adds at most a polylogarithmic cost.  If $\chi<2$, the level-scale
series are geometrically summable; if $\chi=2$, each contributes at most a
power of $L+1$.  Since $\tau_\ell$ is polylogarithmic by assumption, the stated
rates follow.  All registers index levels, degrees of freedom, or
polynomial-precision arithmetic values; the stipulated circuits give the
workspace bound.
\end{proof}

\section{A two-level QLSA viewpoint}
\label{sec:scaled-two-level-qlsa}

A coupled linear system offers another way to express the same correction.
Let $\bm e_\ell=\bm x_\ell-P_\ell\bm x_{\ell-1}$ and define
\begin{equation}
 \bm z_\ell=\begin{bmatrix}h_\ell^{-2}\bm e_\ell\\\bm x_{\ell-1}\end{bmatrix}.
\label{eq:scaled-extended-unknown}
\end{equation}

\Cref{lem:J-L2-scale} gives
  $\|\bm e_\ell\|_2=O(h_\ell^2)$ for bounded
loads.  Galerkin consistency shows that $S_\ell\bm z_\ell=\bm r_\ell$, where
\begin{equation}
 S_\ell=\begin{bmatrix}
 h_\ell^4A_\ell&h_\ell^2A_\ell P_\ell\\
 h_\ell^2P_\ell^\dagger A_\ell&2A_{\ell-1}
 \end{bmatrix},\qquad
 \bm r_\ell=\begin{bmatrix}h_\ell^2\bm b_\ell\\2\bm b_{\ell-1}\end{bmatrix}.
\label{eq:S-r-scaled}
\end{equation}
The linear readout is $[h_\ell^2\bm g_\ell;0]^\dagger\bm z_\ell$ and
the quadratic readout is $\bm z_\ell^\dagger M_\ell^{\rm ext}\bm z_\ell$, with
\begin{equation}
 M_\ell^{\rm ext}=\begin{bmatrix}
 h_\ell^4M_\ell&h_\ell^2M_\ell P_\ell\\
 h_\ell^2P_\ell^\dagger M_\ell&0
 \end{bmatrix}.
\label{eq:M-ext-scaled}
\end{equation}
Their normalizations contain the desired $h_\ell^{2-\chi}$ factor if the
scaled joint response can be implemented efficiently.

\begin{proposition}[Hierarchical reduction and equivalent operator access]
\label{prop:coupled-ritz-reduction}
Let $\widehat T_\ell=[P_\ell\;J_\ell]$ and
$\mathcal U_\ell=\widehat T_\ell\oplus I_{N_{\ell-1}}$.  Then
\begin{equation}
 \mathcal U_\ell^\dagger S_\ell\mathcal U_\ell
 =\begin{bmatrix}
 h_\ell^4A_{\ell-1}&0&h_\ell^2A_{\ell-1}\\
 0&h_\ell^4\Sigma_\ell&0\\
 h_\ell^2A_{\ell-1}&0&2A_{\ell-1}
 \end{bmatrix}.
\label{eq:coupled-system-ritz-form}
\end{equation}
Its solution has the compatible joint response
\begin{equation}
 \bm z_\ell=\mathcal Z_\ell\bm b_\ell,\qquad
 \mathcal Z_\ell=\begin{bmatrix}
 h_\ell^{-2}J_\ell\Sigma_\ell^{-1}J_\ell^\dagger\\
 B_{\ell-1}P_\ell^\dagger
 \end{bmatrix}
 =\begin{bmatrix}h_\ell^{-2}D_\ell\\B_{\ell-1}P_\ell^\dagger\end{bmatrix}.
\label{eq:compatible-joint-solution-map}
\end{equation}
Suppose $B_{\ell-1}P_\ell^\dagger$ has an $O(1)$-normalized block encoding.
Then an $O(1)$-normalized coherent rectangular block encoding of
$\mathcal Z_\ell$ is equivalent, up to constant block-assembly overhead,
to an $O(h_\ell^2)$-normalized block encoding of $D_\ell$.
\end{proposition}

\begin{proof}
Ritz orthogonality gives
$\widehat T_\ell^\dagger A_\ell\widehat T_\ell
=\operatorname{diag}(A_{\ell-1},\Sigma_\ell)$ and hence
\eqref{eq:coupled-system-ritz-form}.  Its transformed right-hand side is
$[h_\ell^2\bm b_{\ell-1};h_\ell^2J_\ell^\dagger\bm b_\ell;
2\bm b_{\ell-1}]$.  The first and third equations set the first coarse
coordinate to zero and the last to $B_{\ell-1}\bm b_{\ell-1}$; the detail
coordinate is $h_\ell^{-2}\Sigma_\ell^{-1}J_\ell^\dagger\bm b_\ell$.
This proves \eqref{eq:compatible-joint-solution-map}.  Projecting its first
row yields the normalized $h_\ell^{-2}D_\ell$ block.  Conversely, a
two-label preparation stacks that block and the supplied coarse response
at $O(1)$ normalization.  Errors in the first row scale by $h_\ell^2$
when translated back to $D_\ell$.
\end{proof}

Thus hierarchical elimination constructs the same Ritz correction and Schur
solve as the direct estimator.  The equivalence concerns coherent access to
the solution operator.  Preparing its normalized output for one particular
load is a weaker requirement and does not imply the converse.

An ideal energy preconditioner illustrates this distinction.  Set
$\mathcal H_\ell=h_\ell^4A_\ell\oplus A_{\ell-1}$ and
$\widetilde{Q}_\ell=A_\ell^{1/2}P_\ell A_{\ell-1}^{-1/2}$, so $\widetilde{Q}_\ell^\dagger \widetilde{Q}_\ell=I$.
Then
\[
 \mathcal H_\ell^{-1/2}S_\ell\mathcal H_\ell^{-1/2}
 =\begin{bmatrix}I&\widetilde{Q}_\ell\\\widetilde{Q}_\ell^\dagger&2I\end{bmatrix}
\]
has spectrum contained in
$\{1,(3-\sqrt5)/2,(3+\sqrt5)/2\}$.  This bounded generalized condition
number is not itself an efficient encoding of $\mathcal Z_\ell$; any
construction that supplies that encoding also supplies the corrected
Green operator by the proposition.  No impossibility claim about other
QLSA formulations follows from this observation.

\section{Conclusions}
\label{sec:summary-discussion}

The corrected Green operator provides a link between multilevel finite element
analysis and quantum output estimation.  Its Ritz-complement factorization
identifies the part of the inverse response added by refinement, while the
quantum construction preserves that correction's small scale before measurement.
Under the stated access requirements, the resulting linear and quadratic
estimators remove polynomial mesh dependence for readout orders up to two.
The continuum accuracy remains governed by the bias of the particular output.

The explicit examples show two realizations of the same mechanism: local
midpoint details in one dimension and frequency-band filters in nested spectral
spaces.  The latter also connect with direct quantum spectral PDE solvers
\cite{HuangAntonioliBarbaresco2026}.  For more general finite element spaces,
the analytic norm estimate is available under dual approximation, but efficient
access to stable Ritz coordinates is an additional algorithmic requirement.
Constructing such access from hierarchical bases, multilevel Riesz bases, or
appropriately related localized decompositions is the main outstanding step.

The two-level QLSA formulation clarifies the role of preconditioning.  A bounded
generalized condition number is useful for the solve, but the circuit must also
preserve the physical scale of the requested output.  Hierarchical elimination
of the coupled system produces exactly the same Ritz-detail inverse that is used
by the direct estimator.

Further work could compare these constructions numerically under common data
and readout models, and investigate mixed formulations in which flux is a
primary variable.  Such a formulation can reduce the readout norm, while
changing the system to be solved and the approximation rate of the relevant
variable.  These tradeoffs reinforce the need to analyze preparation, solution
access, and measurement together.

\section*{Acknowledgments}
This work was supported by NSF Grant DMS-2411120.

\bibliographystyle{amsplain}
\bibliography{ml_pde_refs}

\end{document}